\documentclass[a4paper, 11pt]{article}

\usepackage{fullpage}
\usepackage{times}

\usepackage{amsmath}
\usepackage{amsthm, amssymb}
\usepackage{xcolor}
\usepackage{mathtools}
\usepackage{algorithm2e}
\usepackage[numbers]{natbib}

\usepackage[utf8]{inputenc} % Required for inputting international characters
\usepackage[T5]{fontenc} % Output font encoding for international characters

\usepackage{hyperref}
\hypersetup{colorlinks,
            linkcolor=blue,
            citecolor=blue,
            urlcolor=magenta,
            linktocpage,
            plainpages=false}
\usepackage[capitalise,noabbrev]{cleveref}

\DeclareMathOperator*{\argmax}{arg\,max}
\DeclareMathOperator*{\argmin}{arg\,min}
\DeclarePairedDelimiter\abs{\lvert}{\rvert}
\DeclareMathAlphabet{\pazocal}{OMS}{zplm}{m}{n}
\DeclarePairedDelimiter{\Exp}{\mathbb{E}[}{]}
\DeclarePairedDelimiter{\Prob}{\mathbb{P}[}{]}

\newcommand{\real}{\mathbb{R}}

\newtheorem{definition}{Definition}
\newtheorem*{definition*}{Definition}
\newtheorem{theorem}{Theorem}

\newtheorem{corollary}[theorem]{Corollary}

\title{Online learning for min-max discrete problems}
\date{}
\author{Evripidis Bampis$^1$,  Dimitris Christou$^{1,2}$, Bruno Escoffier$^1$, Nguyễn Kim Thắng$^{1,3}$\\ $ $\\
$^1$ Sorbonne Universit\'e, CNRS, LIP6, F-75005 Paris, France \\
$ $\\ $^2$ National Technical University of Athens, Greece\\ $ $ \\ $^3$ IBISC, Univ Evry, University Paris-Saclay, France}

\begin{document}

\maketitle

%%%%%%%%%%%%%%%%%%%%%%%%%%%%%%%%%%%%%%%%%%%%%%%%%%%%%%%%%%%%%
%%%%%%%%%%%%%%%%%%%%%%%%%  ABSTRACT %%%%%%%%%%%%%%%%%%%%%%%%%
%%%%%%%%%%%%%%%%%%%%%%%%%%%%%%%%%%%%%%%%%%%%%%%%%%%%%%%%%%%%%
\begin{abstract}
We study various discrete nonlinear combinatorial optimization problems in an online learning framework.
In the first part, we address the computational complexity of designing 
%question of whether there are negative results showing that getting a 
vanishing regret (and vanishing approximate regret) algorithms. We provide a general reduction showing that many (min-max) polynomial time solvable problems not only do not have a vanishing regret, but also no vanishing approximation $\alpha$-regret, for some $\alpha$, unless $NP=RP$. Then, we focus on a particular min-max problem, the min-max version of the vertex cover problem which is solvable in polynomial time in the offline case. Our reduction implies that there is no $(2-\epsilon)$-regret online randomized algorithm unless Unique Game is in $RP$. Besides, we prove a matching upper bound providing an online algorithm based on the online gradient descent method. 

In the second part, we turn our attention to online learning algorithms that are based on an offline optimization oracle that, given a set of instances of the problem, is able to compute the optimum static solution.  We start by presenting an online algorithm with vanishing regret that is based on the Follow the Perturbed Leader algorithm for an interesting generalization of knapsack (Barman et al. ICALP 2012). Subsequently, we show that for several min-max (nonlinear) discrete optimization problems,  it is strongly $NP$-hard to solve the offline optimization oracle, even for problems that can be solved in polynomial time in the static case (e.g. min-max vertex cover, min-max perfect matching, etc.). This also provides a useful insight into the connection between the non-linear nature of some problems and the increase of their computational hardness when moved to an online learning setting.
\end{abstract}
\newpage

%%%%%%%%%%%%%%%%%%%%%%%%%%%%%%%%%%%%%%%%%%%%%%%%%%%%%%%%%%%%%
%%%%%%%%%%%%%%%%%%%%%%%  INTRODUCTION %%%%%%%%%%%%%%%%%%%%%%%
%%%%%%%%%%%%%%%%%%%%%%%%%%%%%%%%%%%%%%%%%%%%%%%%%%%%%%%%%%%%%
\section{Introduction}\label{sec:Introduction}

\par Over the past years, online learning has become a very active research field. This is due to the widespread of applications with evolving or adversarial environments, e.g. routing schemes in networks \cite{Awerbuch}, online marketplaces \cite{Cesa-Bianchi}, spam filtering \cite{Hazanothers16:Introduction-to-online}, etc. An online learning algorithm has to iteratively choose an action over a (possible infinite) set of feasible decisions. A loss/reward is associated to each decision which may be adversarially chosen. The losses/rewards are unknown to the algorithm beforehand, making it impossible, in general, to be competitive with the best dynamic strategy that may change decisions each period. On that note, the goal is to minimize the {\em regret}, i.e. the difference between the total loss/reward of the online algorithm and that of the best offline decision, if one had to choose a single decision in hindsight

\par A ``good'' online learning algorithm is an algorithm whose regret is sublinear as a function of the length of the time-horizon since then, on the average, the algorithm performs as well as the best single action in hindsight on the long term. Such an online algorithm is called an online learning algorithm with {\em vanishing regret}. For problems for which the offline version is $NP$-hard, the notions of regret and vanishing regret have been extended to the notions of $\alpha$-\textit{regret} and  $\alpha$-\textit{vanishing regret} in order to take into account the existence of an $\alpha$-approximation algorithm instead of an exact algorithm for solving the offline optimization problem.

\par  While a lot of online learning problems can be modeled as the so called ``experts problem'' by associating a feasible solution to an expert, there is clearly an efficiency challenge since there are potentially an exponential number of solutions making problematic the use of such an approach in practice. Hazan and Koren~\cite{HazanKoren16:The-computational-power} proved that a \textit{vanishing regret} algorithm  with running-time polynomial in the size of the problem does not exist in general settings without any assumption on the structure of the problem. However, there are many results that prove that the structure of a problem can be indeed exploited in order to get such an algorithm. Notably, Kalai and Vempala \cite{KalaiVempala} proved that for linear problems that are %not strongly $NP$-hard
polynomially solvable, such an algorithm can always be constructed by using any algorithm for the offline problem as an oracle. Other methods have been used such as the online gradient descent \cite{Zinkevich} for convex objectives, a generalization of the Follow the Perturbed Leader algorithm  for submodular objective functions \cite{HazanKale12:Online-submodular} and the Generalized Follow the Perturbed Leader \cite{DudikHaghtalab17:Oracle-efficient-online} algorithm for general objectives.

\par Our work takes into account the computational efficiency of online learning algorithms in the same vein as the works in \cite{AgarwalGonen19:Learning-in-Non-convex,KalaiVempala,HazanKale12:Online-submodular,SyrgkanisKrishnamurthy16:Efficient-algorithms,DaskalakisSyrgkanis16:Learning-in-auctions:,DudikHaghtalab17:Oracle-efficient-online,KakadeKalai09:Playing-games,Garber17:Efficient-online}.  We study various discrete {\it nonlinear} combinatorial optimization problems in an online learning framework. We focus in particular on min-max (discrete) optimization problems, for two main reasons. First, together with the \emph{utilitarian} objective, the \emph{egalitarian} objective --- represented by the min-max function --- is arguably one of the most studied families of objective functions in Optimization, Operations Research, Game Theory, etc (see for instance \cite{GabowT88,berman,Camerini78}). Second and most importantly with the current article, as we will see the non-linearity of min-max problems lead to several hardness results in the online learning framework, even for min-max problems which are trivial to solve
efficiently in the offline setting, thus exhibiting an interesting difference with linear discrete problems.

\par While most works focus on analyzing online learning algorithms  that can achieve vanishing regret or $\alpha$-vanishing regret for certain families of problems, not many results that show the computational hardness of achieving those properties are known. 
Note that, a hardness result for offline optimization do not necessarily translate to an analogous
hardness result for online learning because an online learning algorithm could conceivably achieve a small cost with different solutions at different times, while never figuring out the best fixed solution in hindsight.
In generally, our goal in this work is to address the three following  central questions:
\begin{itemize}
    \item $(Q1)$: Are there  negative results showing that getting vanishing regret (or even vanishing approximate regret) is computationally hard?
    
    \item $(Q2)$: How does the non-linear nature of an online learning problem tie in with its computational hardness?
    
    \item $(Q3)$: Are there some notable differences in the efficiencies of \textit{Follow the Leader} and \textit{gradient descent} strategies for discrete problems?
\end{itemize}

%%%%%%%%%%%%%%%%%%%%%%%%%%%%%%%%%%%%%%%%%%%%%%%%%%%%%%%%%%%%%
%%%%%%%%%%%%%%%%%%%%%%%%%  FRAMEWORK %%%%%%%%%%%%%%%%%%%%%%%%
%%%%%%%%%%%%%%%%%%%%%%%%%%%%%%%%%%%%%%%%%%%%%%%%%%%%%%%%%%%%%

\subsection{The online learning framework}

\par An online learning problem consists of a decision-space $\pazocal{X}$, a state-space $\pazocal{Y}$ and an objective function $f:\pazocal{X}\times\pazocal{Y}\rightarrow\real$ that can be either a cost or a reward function. Any problem of this class can be viewed as an iterative adversarial game with $T$ rounds where the following procedure is repeated for $t=1,\dots,T$:\textit{ (a)} the algorithm first chooses an action $x^t\in\pazocal{X}$; \textit{(b)} after the algorithm has committed to its choice, the nature reveals a state $y^t\in\pazocal{Y}$; \textit{(c)} the algorithm observe the state $y^t$ and suffers a loss (or gains a reward) $f^t(x^t) = f(x^t,y^t)$. We use $f^t(x)$ as another way to refer to the objective function \textit{f} after observing the state $y^t$, i.e. the objective function at round $t$.

\par The objective of the player is to minimize/maximize the accumulative cost/reward of his decided actions, which is given by the aggregation $\sum_{t=1}^Tf(x^t,y^t)$. An online learning algorithm is any algorithm that decides the actions $x^t$ at every round before observing $y^t$. We compare the decisions $(x^1,\dots,x^T)$ of the algorithm with those of the best \textit{static} action in hindsight, defined as:
$x^* =  \argmin_{x\in\pazocal{X}}\sum_{t=1}^Tf(x,y^t)$, or $x^* = \argmax_{x\in\pazocal{X}}\sum_{t=1}^Tf(x,y^t)$,
for minimization or maximization problems respectively. This is the action that a (hypothetical) offline oracle would compute, if it had access to the entire sequence $y^1,\dots, y^T$. The typical measurement for the efficiency of an online learning algorithm is the \textit{regret}, defined as: 

$$
R_T = \sum_{t=1}^Tf(x^t,y^t) - \sum_{t=1}^Tf(x^*,y^t)
$$

\par As generally no deterministic learning algorithm has a good worst case behaviour, a learning algorithm typically uses some kind of randomness, and the regret denotes  the expectation of the above quantity; in all the article, by referring to a learning algorithm we mean a randomized learning algorithm (including for the impossibility results).
We are interested in online learning algorithms that have the \textit{"vanishing regret"} property. This means that as the "game" progresses ($T\rightarrow\infty$), the average deviation between the algorithm's average cost/payoff to the average cost/payoff of the optimum action in hindsight tends to zero. Typically, a vanishing regret algorithm is an algorithm with regret $R_T$ such that: $\lim_{T\rightarrow\infty}\frac{R_T}{T} = 0$.
However, as we are interested in polynomial time algorithms, we consider only vanishing regret $R_T = O(T^c)$ where $0\leq c <1$ (that guarantees the convergence in polynomial time). Throughout the paper, whenever we mention vanishing regret, we mean regret $R_T = O(T^c)$ where $0\leq c <1$.

\par For many online learning problems, even their offline versions are $NP$-hard. Thus, it is not feasible to produce a vanishing regret sequence with an efficient algorithm. For such cases, the notion of $\alpha$-{\em regret} has been defined as: 

$$ R_T^{\alpha} =  \sum_{t=1}^Tf(x^t,y^t) - \alpha \sum_{t=1}^Tf(x^*,y^t). $$
Hence, we are interested in vanishing $\alpha$-regret sequences for some $\alpha$ for which we know how to approximate the offline problem. The notion of vanishing $\alpha$-regret is defined in the same way as that of vanishing regret. In this article we focus on computational issues. Efficiency for an online learning algorithm needs to capture both the computation of $x^t$ and the convergence speed. This is formalized in the following definition (where $n$ denotes the size of the instance).

\begin{definition}
A {\it polynomial time vanishing $\alpha$-regret algorithm} is an online learning algorithm for which (1) the computation of $x^t$ is polynomial in $n$ and $t$ (2) the expected $\alpha$-regret is bounded by $p(n)T^{c}$ for some polynomial $p$ and some constant $0\leq c<1$. We note that $n$ is the size of the instance which can be much smaller than the decision space of the problem.
\end{definition}

\noindent Note that in case $\alpha=1$, we simply use the term \textit{polynomial time vanishing regret} algorithm.

%%%%%%%%%%%%%%%%%%%%%%%%%%%%%%%%%%%%%%%%%%%%%%%%%%%%%%%%%%%%%
%%%%%%%%%%%%%%%%%%%%%%%  CONTRIBUTION %%%%%%%%%%%%%%%%%%%%%%%
%%%%%%%%%%%%%%%%%%%%%%%%%%%%%%%%%%%%%%%%%%%%%%%%%%%%%%%%%%%%%

\subsection{Our contribution}

\par In Section 2, we provide a general reduction showing that many (min-max) polynomial time solvable problems not only do not have a vanishing regret, but also no vanishing approximation $\alpha$-regret, for some $\alpha$, unless $NP=RP$. The result holds for  online randomized algorithms against an oblivious adversary. This result gives an answer to question $(Q1)$ that was stated above. Then, we focus on a particular min-max problem, the min-max version of the vertex cover problem which is solvable in polynomial time in the offline case. The previous reduction proves that there is no $(2-\epsilon)$-regret online randomized algorithm unless Unique Game is in $RP$. Besides, we prove a matching upper bound providing an online algorithm based on the online gradient descent method. 

\par In Section 3, we turn our attention to online learning algorithms that are based on an offline optimization oracle that, given a set of instances of the problem, is able to compute the optimum static solution. This approach is known in the literature as the \textit{Follow the Leader} method. To the best of our knowledge, up to now algorithms based on the Follow the Leader method for non-linear objective functions require an exact oracle or a constant additive error oracle in order to obtain vanishing regret. By slightly extending the \textit{Generalized Follow the Perturbed Leader} algorithm of Dudik et al. \cite{DudikHaghtalab17:Oracle-efficient-online}, we are able to replace the assumption of a constant additive error oracle with that of a FPTAS oracle. We use this extension to present an online algorithm with vanishing regret that is based on the Follow the Perturbed Leader algorithm for a well-studied generalization of the knapsack problem with applications in different areas like cloud computing or connection management in wireless access points~\cite{Antoniadis, BarmanUCM12}.

Subsequently, in Section 3.1 we show strong NP-hardness results for the multiple instance version of some offline problems, which indicates that follow-the-leader-type strategies {\it cannot} be used for the online learning problem, at least with our current knowledge. More precisely, we show that such hardness results hold even for problems that can be solved in polynomial time in the static case (such as min-max vertex cover or min-max perfect matching). In particular, we also prove that the offline optimization oracle is strongly $NP$-hard for the problem of scheduling a set of jobs on $m$ identical machines, where $m$ is a fixed constant. These results provide some useful insight towards answering question $(Q2)$. 

Furthermore, the hardness results on the oracle for the min-max vertex cover problem, paired with the upper bound that is presented in Section 2.2 using the online gradient descent method, provide some answer to question $(Q3)$.

%%%%%%%%%%%%%%%%%%%%%%%%%%%%%%%%%%%%%%%%%%%%%%%%%%%%%%%%%%%%%
%%%%%%%%%%%%%%%%%%%%%%%  RELATED WORK %%%%%%%%%%%%%%%%%%%%%%%
%%%%%%%%%%%%%%%%%%%%%%%%%%%%%%%%%%%%%%%%%%%%%%%%%%%%%%%%%%%%%

\subsection{Further related works}
Online Learning, or Online Convex Optimization, is an active research domain. In this section, 
we only summarize works which are directly related to ours. We refer the reader to comprehensive books 
\cite{Shalev-Shwartzothers12:Online-learning,Hazanothers16:Introduction-to-online}
and references therein for a more complete overview. The first vanishing regret algorithm has been given by Hannan~\cite{Hannan57:Approximation-to-Bayes}. 
Subsequently, Littlestone and Warmuth \cite{LittlestoneWarmuth94:The-weighted-majority} and Freund and Schapire~\cite{FreundSchapire97:A-decision-theoretic-generalization} gave improved algorithms with regret $\sqrt{\log(|\mathcal{A}|)} o(T)$ where $|\mathcal{A}|$ is the size of the action space. However, these algorithms have running-time $\Omega({|\mathcal{A}|})$ 
which is exponential in the size of the input for many applications, in particular for combinatorial optimization problems. 
An intriguing question is whether there exists a vanishing regret online algorithm with running-time polynomial in $\log(|\mathcal{A}|)$. 
Hazan and Koren~\cite{HazanKoren16:The-computational-power} proved that no such 
algorithm exists in general settings without any assumption on the structure. 
Designing online polynomial-time algorithms with approximation and vanishing regret
guarantees for combinatorial optimization problems is a major research agenda. 

In their breakthrough paper, Kalai and Vempala~\cite{KalaiVempala} presented 
the first efficient online algorithm, called \emph{Follow the Perturbed Leader} (FTPL), 
for linear objective functions. The strategy consists of adding perturbation to the cumulative gain (payoff) of each action and then selecting the action with the highest perturbed gain. This strategy has been generalized and 
successfully applied to several settings 
\cite{HazanKale12:Online-submodular,SyrgkanisKrishnamurthy16:Efficient-algorithms,DaskalakisSyrgkanis16:Learning-in-auctions:,DudikHaghtalab17:Oracle-efficient-online}. 
Specifically, FTPL and its generalized versions have been used to design efficient online vanishing regret algorithms
with oracles beyond linear settings: to submodular settings \cite{HazanKale12:Online-submodular} and
non-convex settings \cite{AgarwalGonen19:Learning-in-Non-convex}. 
However, all these approaches require best-response oracles, and as we show in the current paper, for several problems such best-response oracles require
exponential time computation.

Another direction is to design online learning algorithms using (offline polynomial-time) 
approximation algorithms as oracles. Kakade et al.\cite{KakadeKalai09:Playing-games} provided an algorithm which is inspired by Zinkevich's algorithm \cite{Zinkevich} (gradient descent): at every step, the algorithm updates the current solution in the direction of the gradient and project back to the feasible set using an approximation algorithm.  
They showed that given an $\alpha$-approximation algorithm for a \emph{linear} optimization problem,
after $T$ prediction rounds (time steps)
the online algorithm achieves an $\alpha$-regret bound of $O(T^{-1/2})$ using $T$ calls to the approximation 
algorithm per round in average. Later on, Garber \cite{Garber17:Efficient-online} 
gave an algorithm with $\alpha$-regret bound 
of $O(T^{-1/3})$ using only $O(\log T)$ calls to the approximation algorithm per round in average. 
These algorithms rely crucially on the linearity of the objective functions and it remains an interesting 
open question to design algorithms for online non-linear optimization problems.

%%%%%%%%%%%%%%%%%%%%%%%%%%%%%%%%%%%%%%%%%%%%%%%%%%%%%%%%%%%%%
%%%%%%%%%%%%%%%  SECTION 2: MINMAX LEARNING %%%%%%%%%%%%%%%%%
%%%%%%%%%%%%%%%%%%%%%%%%%%%%%%%%%%%%%%%%%%%%%%%%%%%%%%%%%%%%%

\section{Hardness of online learning for min-max problems}\label{sec:Hardness}

%%%%%%%%%%%%%%%%%%%%%%%%%%%%%%%%%%%%%%%%%%%%%%%%%%%%%%%%%%%%%
%%%%%%%%%%%%%%%%%%%%%%%%%%  REDUCTION %%%%%%%%%%%%%%%%%%%%%%%
%%%%%%%%%%%%%%%%%%%%%%%%%%%%%%%%%%%%%%%%%%%%%%%%%%%%%%%%%%%%%
\subsection{General reduction}

As mentioned in the introduction, in this section we give some answers to question $(Q1)$ on ruling out the existence of vanishing regret algorithm for a broad family of online min-max problems, even for ones that are polynomial-time solvable in the offline case. 
In fact, we provide a general reduction (see Theorem~\ref{th2}) showing that many min-max problems do not admit vanishing $\alpha$-regret for some $\alpha > 1$ unless $NP=RP$. 

More precisely, we focus on a class of \emph{cardinality minimization problems} where, given an $n$-elements set $\pazocal{U}$, a set of constraints $\pazocal{C}$ on the subsets of $\pazocal{U}$ (defining feasible solutions) and an integer $k$, the goal is to determine whether there exists a feasible solution of size at most $k$. This is a general class of problems, including for instance graph problems such as Vertex Cover, Dominating Set, Feedback Vertex Set, etc.

Given such a cardinality problem $\pazocal{P}$, let \textit{min-max-}$\pazocal{P}$ be the optimization problem where given non-negative weights for all the elements of $\pazocal{U}$, one has to compute a feasible solution (under the same set of constraints $\pazocal{C}$ as in problem $\pazocal{P}$) such that the maximum weight of all its elements is minimized.
The \textit{online min-max-}$\pazocal{P}$ problem is the online learning variant of min-max-$\pazocal{P}$, where the set of elements $\pazocal{U}$ and the set of constrains $\pazocal{C}$ remain static but the weights on the elements of $\pazocal{U}$ change over time.

Interestingly, the min-max versions of all the problems mentioned above are {\it polynomially solvable}. This is actually true as soon as, for problem $\pazocal{P}$, every superset of a feasible solution is feasible. Then one just has to check for each possible weight $w$ if the set of all elements of weight at most $w$ agrees with the constraints.
For example, one can decide if there exists a vertex cover with the maximum weight $w$ as follows: 
remove all vertices of weight strictly larger than $w$, and check if the remaining vertices form a vertex cover. 

We will show that, in contrast, if $\pazocal{P}$ is $NP$-complete then its online learning min-max version has no vanishing regret algorithm (unless $NP=RP$), and that if $\pazocal{P}$ has an inapproximability gap $r$, then there is no vanishing $\alpha$-regret algorithm for its online learning min-max version, for any $\alpha < r$. Let us first recall the notion of approximation gap, where $x_{opt}$ denotes a feasible solution of minimum size to the cardinality problem $\mathcal{P}$.

\begin{definition}
Given two numbers $0\leq A<B \leq 1$, let \textit{[A,B]-Gap-}$\pazocal{P}$ be the decision problem where given an instance of $\pazocal{P}$ such that $|x_{opt}| \leq An$ or $|x_{opt}| \geq Bn$, we need to decide whether $|x_{opt}| < Bn$.
\end{definition}

Now we can state the main result of the section.

\begin{theorem}\label{th2} 
Let $\pazocal{P}$ be a cardinality minimization problem and $A,B$ be real numbers with $0\leq A<B\leq 1$. Assume that the problem \textit{[A,B]-Gap-}$\pazocal{P}$ is $NP$-complete. 
Then, for every $\alpha < \frac{B}{A}$  there is no randomized polynomial-time vanishing $\alpha$-regret algorithm for online min-max-$\pazocal{P}$ unless $NP=RP$.
\end{theorem}

%{\color{blue}
\begin{proof}
We prove this theorem by deriving a polytime algorithm for \textit{[A,B]-Gap-}$\pazocal{P}$ that gives, under the assumption of a vanishing $\alpha$-regret algorithm for online min-max-$\pazocal{P}$ with $\alpha < \frac{B}{A}$, the correct answer with probability of error at most $D$ for some constant $D<1$  if $|x_{opt}|\leq An$, and with no error if $|x_{opt}|\geq Bn$. This would imply that the \textit{[A,B]-Gap-}$\pazocal{P}$ problem is in $RP$ and thus $NP=RP$. 

Let $\mathcal{O}$ be a (randomized) vanishing $\alpha$-regret algorithm for online min-max-$\pazocal{P}$ 
for some $\alpha = \left( \frac{B}{A} - \epsilon \right) = (1 - \epsilon') \frac{B}{A}$ where $\epsilon > 0$ is a constant and $\epsilon' = \frac{A}{B} \epsilon$. 
Let $T$ be a time horizon which will be fixed later. We construct the following (offline) randomized algorithm for 
\textit{[A,B]-Gap-}$\pazocal{P}$ using $\mathcal{O}$ as an oracle (subroutine). 
At every step $1 \leq t \leq T$, the algorithm:
\begin{itemize}\item uses the oracle $\mathcal{O}$ to compute a solution $x^{t}$ (i.e., it draws $x^t$ from the random distribution $X_t$ given by the oracle),\item  fixes the weights randomly as follows: it chooses one element uniformly at random in $\mathcal{U}$ and assign weight 1 to that element, and weight 0 to all other ones (so each element has weight 1 with probability $1/n$).
\end{itemize}%Consequently, the cost incurred to $\mathcal{O}$ is 1 at every step. 
Note that these weights are chosen in an oblivious manner (independently on the solutions computed by the oracle $\pazocal{O}$).  %Intuitively, the assignments will be used 
%to learn about the optimal solution of the \textit{[A,B]-Gap-}$\pazocal{P}$ problem 
%(given the performance of the learning algorithm $\mathcal{O}$). 
The formal description is given in Algorithm~\ref{alg:ReductionAlgorithm}.

\RestyleAlgo{boxruled}
\LinesNumbered
\begin{algorithm}[ht]
\caption{Algorithm  for the \textit{[A,B]-Gap-}$\pazocal{P}$ problem \label{alg:ReductionAlgorithm}}
\For{ $t=1,2,\dots, T$}{
  Choose $x^t \in \pazocal{X}$ according to the random distribution $X_t$ given by algorithm $\mathcal{O}$.\\
  \textbf{if} $|x^t| < Bn$ \textbf{then} \Return{\texttt{Yes}, i.e., $|x_{opt}| \leq An$}.\\
  Assign weight 1 to an element of $\pazocal{U}$ chosen uniformly at random and 0 to all other elements of $\pazocal{U}$. \\
  Feed the weight vector and the cost $f^t(x^t) = \max_{u\in x^t}w^t(u)$ back to $\mathcal{O}$. \\
}
\Return{\texttt{No}, i.e., $|x_{opt}| \geq Bn$}.
\end{algorithm}

%\noindent After the game ends, if $|x^t| < Bn$ for some $t$, then we answer that $|x_{opt}| \leq An$, else we answer $|x_{opt}| \geq Bn$. In the remainder of the proof we bound the probability of error of this answer.

We now analyze the probability of error of Algorithm~\ref{alg:ReductionAlgorithm}.

If $|x_{opt}|\geq Bn$, then obviously all the solutions $x^t$ computed by $\mathcal{O}$ have sizes at least $Bn$, so the algorithm return \texttt{No}, with no error.

In the remaining of the proof, assume that $|x_{opt}|\leq An$. We now analyse the probability that Algorithm 1 fails (returns No) in this case, i.e., conditionally to the fact that $|x_{opt}|\leq An$ (all the probabilities in the sequel are conditional to $|x_{opt}|\leq An$).

Let $E$ denote the event that the algorithm returns a wrong answer (when $|x_{opt}|\leq An$). The algorithm fails iff  at each time step $|x^t|\geq Bn$. We get:
$$
\Prob{E}= \mathbb{P} \left [  \cap_{t=1}^{T} \{ |X_t|\geq Bn \} \right] \leq \Prob{X\geq TBn}
$$
where $X=\sum_{t=1}^{T} |X_t|$ is the random variable equal to the sum of sizes of (random) solutions outputs by $\mathcal{O}$. As $X\geq 0$, $\Prob{X\geq TBn}\leq \frac{\Exp{X}}{TBn}=\frac{\sum_{t=1}^T \Exp{|X_t|}}{TBn}$.
 Since the element of weight 1 at time $t$ is chosen uniformly at random, $\Exp{f^t(X_t)}= \Exp{|X_t|}/n$, so
$$
\Prob{E}\leq \frac{\Exp{X}}{TBn} \leq \frac{\sum_{t=1}^T \Exp{f^t(X_t)}}{TB}.
$$

Now, since $|x_{opt}|\leq An$, and since at each time $t$ there is only one element of weight 1 picked uniformly at random, $f^t(x_{opt})=1$ with probability at most $A$, we get
$
\sum_{t=1}^T  \Exp{f^t(x_{opt})} \leq AT.
$
Besides, as $\mathcal{O}$ is a vanishing $\alpha$-regret algorithm with $\alpha = (1 - \epsilon') \frac{B}{A}$, we have that:
$$
\sum_{t=1}^T \Exp{f^t(X_t)}\leq \alpha \sum_{t=1}^T \Exp{f^t(x_{opt})} + p(n)T^c\leq (1-\epsilon')BT+p(n)T^c.
$$
Hence, we deduce that 
$$ 
\Prob{E} \leq \frac{(1-\epsilon')BT+p(n)T^c}{BT}=1-\epsilon'+\frac{p(n)T^{c-1}}{B}.
$$
Choose parameter $T = \left(\frac{ B\epsilon' }{2p(n)}\right)^\frac{1}{c-1} = \left(\frac{ A \epsilon }{2p(n)B }\right)^\frac{1}{c-1}$, 
we get that $\Prob{E} \leq 1- \frac{\epsilon'}{2} = 1- \frac{A \epsilon}{2B}$, 
which is a constant strictly smaller than 1. Besides, the running time of Algorithm~\ref{alg:ReductionAlgorithm} is 
polynomial since it consists of $T$ (polynomial in the size of the problem) iterations and
the running time of each iteration is polynomial (as $\mathcal{O}$ is a polynomial time algorithm).  

In conclusion, if there exists a vanishing $\alpha$-regret algorithm for online min-max-$\pazocal{P}$,
then the $NP$-complete problem \textit{[A,B]-Gap}-$\pazocal{P}$ is in $RP$, implying $NP=RP$.
\end{proof}

%}

The inapproximability (gap) results for the aforementioned problems give lower bounds on the approximation ratio $\alpha$ of any vanishing $\alpha$-regret algorithm for their online min-max version. Note that these lower bounds hold with an oblivious adversary (weights are chosen in a non adaptive way), and that it applies to many min-max problems that are (trivially) polynomially solvable. 

For instance, the online min-max dominating set problem has no vanishing constant-regret algorithm based on the approximation hardness in~\cite{RazS97}. We state the lower bound explicitly for the online min-max vertex cover problem in the following corollary, as we refer to it later by showing a matching upper bound. They are based on the hardness results for vertex cover in \cite{KhotMS18} and \cite{KhotR08} ($NP$-hardness and UGC-hardness, respectively).

\begin{corollary}   \label{cor:min-max-VC-lower}
The online min-max vertex cover problem does not admit a polynomial time vanishing $(\sqrt{2}-\epsilon)$-regret unless $NP=RP$. It does not admit a polynomial time vanishing $(2-\epsilon)$-regret unless Unique Game is in $RP$.
\end{corollary}

Now, consider $NP$-complete cardinality problems which have no known inapproximability gap (for instance Vertex Cover in planar graphs, which admits a PTAS). Then we can show the following impossibility result (see Appendix~\ref{app:coro} for the proof).

\begin{corollary} \label{cor:no-no-regret}
If a cardinality problem $\pazocal{P}$ is $NP$-complete, then there is no vanishing regret algorithm for online min-max-$\pazocal{P}$ unless $NP=RP$.
\end{corollary}

%%%%%%%%%%%%%%%%%%%%%%%%%%%%%%%%%%%%%%%%%%%%%%%%%%%%%%%%%%%%%
%%%%%%%%%%%%%%%%%%%  MINMAX VERTEX COVER %%%%%%%%%%%%%%%%%%%%
%%%%%%%%%%%%%%%%%%%%%%%%%%%%%%%%%%%%%%%%%%%%%%%%%%%%%%%%%%%%%

\subsection{Min-max Vertex Cover: matching upper bound with Gradient Descent}\label{sec:minmaxvertexcover}

\par In this section we will present an online algorithm for the min-max vertex cover problem based on the 
classic \textit{Online Gradient Descent} (OGD) algorithm. In the latter, at every step the solution is obtained by 
updating the previous one in the direction of the (sub-)gradient of the objective and projecting to a feasible 
convex set. The particular nature of the min-max vertex cover problem is that the objective function is the 
$l_\infty$ norm and the set of feasible solutions (all subsets of the vertex set $V$ that are vertex covers) $\pazocal{X}$ is discrete (non-convex). 
In our algorithm, we consider the following standard relaxation of the problem:
$$
\min \max_{i\in V}  ~w_{i} x_{i} \quad
\text{s.t.} \quad x \in \mathcal{Q}: \quad  x_i + x_j \geq 1 ~ \forall (i,j)\in E, \qquad
     0 \leq x_i \leq 1 ~ \forall i\in V.
$$

At time step $t$,  after deciding on a (fractional) action $x^t\in \pazocal{Q}$ we update the solution by a sub-gradient 
$g^t(x^t) = [0,\dots, 0,w_i^t,0,\dots,0]$
with $w_i^t$ in coordinate $i^t(x^t) = \argmax_{1\leq i\leq n}w_i^tx_i^t$
and 0 in other coordinates. 
Moreover, after projecting the solution to the polytope $\mathcal{Q}$, we round the solution by a simple procedure:
if $x^{t+1}_{i} \geq 1/2$ then $X^{t+1}_{i} = 1$ and $X^{t+1}_{i} = 0$ otherwise. 
The formal algorithm is given in Algorithm~\ref{alg:mmvc:MinMaxVC}. 

\RestyleAlgo{boxruled}
\LinesNumbered
\begin{algorithm}[ht]
  \caption{OGD-based algorithm for Online MinMax Vertex Cover\label{alg:mmvc:MinMaxVC}}
 Select an arbitrary fractional vertex cover $x^1\in \mathcal{Q}$ .\\
 \For{ $t=1,2,\ldots$}{
 Round $x^{t}$ to $X^{t}$: $X^{t}_{i} = 1$ if $x^{t}_{i} \geq 1/2$ and $X^{t}_{i} = 0$ otherwise. \\
  Play $X^t \in \{0,1\}^{n}$. 
  Observe $w^t$ (weights of vertices) and incur the cost $g^t(X^t) = \max_{i} w^t_i X^{t}_{i}$.\\
  Update $y^{t+1} =  x^t - \frac{1}{\sqrt{t}}g^t(x^{t}) $.\\
  Project $y^{t+1}$ to $\mathcal{Q}$ w.r.t the $\ell_{2}$-norm: 
  	$x^{t+1} = \text{Proj}_{\mathcal{Q}} \bigl( y^{t+1} \bigr) := \arg \min_{x\in \mathcal{Q} } \| y^{t+1} - x \|_{2}$. 
  
 }
\end{algorithm}

The following theorem (see Appendix~\ref{app:thgd} for the proof), coupled with Corollary~\ref{cor:min-max-VC-lower}, show the tight bound of $2$ on the approximation ratio of polynomial-time online algorithms for Min-max Vertex Cover 
(assuming UGC conjecture).

\begin{theorem}\label{th:gd}
Assume that $W = \max_{1 \leq t \leq T} \max_{1 \leq i \leq n} w^t_i$ is the maximum assigned weight. Then,
after $T$ time steps, Algorithm~\ref{alg:mmvc:MinMaxVC} achieves
$$
\sum_{t=1}^{t} \max_{1 \leq i \leq n} w^t_i X^{t}_{i}
\leq 2 \cdot \min_{X^{*} \in \pazocal{X}} \sum_{t=1}^{t} \max_{1 \leq i \leq n} w^t_i X^{*}_{i} + 
	3 W \sqrt{nT}
$$
\end{theorem}

%%%%%%%%%%%%%%%%%%%%%%%%%%%%%%%%%%%%%%%%%%%%%%%%%%%%%%%%%%%%%
%%%%%%%%%%%%%%%%%%%%%%%%  FTL ISSUES %%%%%%%%%%%%%%%%%%%%%%%%
%%%%%%%%%%%%%%%%%%%%%%%%%%%%%%%%%%%%%%%%%%%%%%%%%%%%%%%%%%%%%

\section{Computational issues for Follow the Leader based methods}\label{sec:FTL}

\par In the previous section, we proved that a large family of online-learning problems does not admit a vanishing regret algorithm (Corollary \ref{cor:no-no-regret}) even though they can be easily solved in polynomial time in the offline case. The purpose of this section is to further explore the role of non-linearity in the hardness gap between an offline problem and the its corresponding online-learning variant. In order to do so,  we focus our attention to the family of \textit{Follow the Leader (FTL)} algorithms that has been extensively studied and used for the online-learning setting.

\par We begin by providing the reader with a short introduction to the \textit{Follow the Leader} based methods. The most natural approach in online learning is for the player to always pick the \textit{leading action}, i.e. the action $x^t$ that is optimal to the observed history $y^1,\dots, y^{t-1}$. However it can be proven (\cite{KalaiVempala}) that any deterministic algorithm that always decides on the leading action can be ``tricked'' by the adversary in order to make decision that are worse than the optimal action in hindsight, thus leading to large regret algorithms. On this regard, we need to add a regularization term containing randomness to the optimization oracle in order to make our algorithms less predictable and more stable. Thus, the \textit{Follow the Regularized Leader} strategy in a minimization problem, consists of deciding on an action $x^t$ such that:

$$
x^t = \argmin_{x\in\pazocal{X}}\left(\sum_{\tau=1}^{t-1}f(x,y^\tau) + R(x)\right)
$$
where $R(x)$ is the regularization term.

\par There are many variations of the Follow the Regularized Leader (FTRL) algorithm that differentiate on the applied objective functions and the type of regularization term. For linear objectives, Kalai and Vempala \cite{KalaiVempala} suggested the Follow the Perturbed Leader algorithm where the regularization term is simply the cost/payoff of each action on a randomly generated instance of the problem. Dudik et al. \cite{DudikHaghtalab17:Oracle-efficient-online} were able to generalize the FTPL algorithm of Kalai and Vempala~\cite{KalaiVempala} for non-linear objectives, by introducing the concept of \textit{shared randomness} and a much more complex perturbation mechanism.

\par A common element between every Follow the Leader based method, is the need for an optimization oracle over the observed history of the problem. This is a minimum requirement since the regularization term can make determining the leader even harder, but most algorithms are able to map the perturbations to the value of the objective function on a set of instances of the problem and thus eliminate this extra complexity. In particular, Kalai and Vempala \cite{KalaiVempala} showed that for every linear objective, if the offline version of the problem can be solved exactly or it admits a FPTAS, then the same holds for the multi-instance offline version of the problem (solved by the oracle) which leads to a vanishing regret algorithm for the online-learning variant of the problem.

\par On the same note, Dudik et al. \cite{DudikHaghtalab17:Oracle-efficient-online} introduced the \textit{Generalized Follow the Perturbed Leader} (GFTPL) algorithm. While analyzing the algorithm, they showed that an exact or an additive error oracle that solves the multi-instance problem is sufficient to achieve the same regret bounds under some extra mild assumptions, even for non-linear objectives. In this section, we will turn our focus to the existence of such oracles that can be efficiently computed. 

\par If the problem solved by the oracle is $NP$-hard, having an efficient algorithm that solves it with any additive error $\epsilon$ is quite improbable. We remark that the assumption of an additive error $\epsilon$ can be replaced by the assumption of the existence of a FPTAS for the oracle. This is captured by the following Theorem:

\begin{theorem}\label{th:FPTAS} 
The upper bounds of the Generalized Follow the Perturbed Leader algorithm hold even under the assumption that there exists a FPTAS oracle instead of an additive error oracle.
\end{theorem}

\begin{proof}
While the proof relies on a single technical observation, it requires to establish the full framework of the GFTPL algorithm. On that note, it is moved to the Appendix \ref{app:a2} and is presented after a (brief) introduction to the GFTPL algorithm (Appendix \ref{app:a1}).
\end{proof}

%%%%%%%%%%%%%%%%%%%%%%%%%%%%%%%%%%%%%%%%%%%%%%%%%%%%%%%%%%%%%
%%%%%%%%%%%%%%%%%%%%%%%%%  KNAPSACK %%%%%%%%%%%%%%%%%%%%%%%%%
%%%%%%%%%%%%%%%%%%%%%%%%%%%%%%%%%%%%%%%%%%%%%%%%%%%%%%%%%%%%%

%%%%%%%%% ATTEMPT BRUNO %%%%%%%%%%%%%%%%%%%%%%%%%%%%

We now give a concrete example of the way that GFTPL algorithm can be applied for non-linear objectives in discrete problems using an FPTAS oracle. Let us consider the {\it generalized knapsack problem} where we are allowed to exceed the capacity of the knapsack but then pay a penalty proportional to the excess weight. Formally: 

\begin{definition}[Generalized Knapsack Problem (GKP)]
Given a set of items $i=1,2,...,n$ with non-negative weights $w_i$ and non-negative profits $p_i$, a knapsack capacity $B\in R_+$ and a constant $c\in R_+$, determine a set of items $A\subseteq [n]$ that maximizes the total profit
$\texttt{profit}(A) = \sum_{i\in A}p_i - c\max\{0, \sum_{i\in A}w_i - B \}$.
\end{definition}

This problem, as well as generalizations with other penalty costs for overweight, have been studied for instance in \cite{BarmanUCM12, Antoniadis} (see there for practical motivations). In an online learning setting, we assume that the capacity of the knapsack and the profit of objects may change over time (while weights $w_i$ and a the constant $c$ are static). 

Since the problem is not linear, we use the the \textit{generalized FTPL} (GFTPL) framework of Dudik et al. \cite{DudikHaghtalab17:Oracle-efficient-online}, which uses a much more complex perturbation mechanism than the ``extra'' random observation in FTPL.  

\begin{theorem}\label{th7}
There is a polynomial time vanishing regret algorithm for GKP.
\end{theorem}
The proof is based on showing that (1) the GFTPL method can be applied to GKP (see appendix~\ref{app:a3}) and (2) there exists a FPTAS for the multi-instance oracle (see appendix~\ref{gkp:fptas}). Then Theorem~\ref{th:FPTAS} applies.

\par The proof of Theorem \ref{th7} also leads to an interesting observation; although we were able to acquire a vanishing regret algorithm for GKP which is non-linear, we were ``forced'' to solve a more general problem (that of the convex knapsack) in order to get the desired oracle. This indicates that the non-linear nature of the problem lead to a significant increase in its online-learning version, offering some useful insight towards question $(Q2)$.

%%%%%%%%%%%%%%%%%%%%%%%%%%%%%%%%%%%%%%%%%%%%%%%%%%%%%%%%%%%%%
%%%%%%%%%%%%%%%%%%%%  HARDNESS REDUCTIONS %%%%%%%%%%%%%%%%%%%
%%%%%%%%%%%%%%%%%%%%%%%%%%%%%%%%%%%%%%%%%%%%%%%%%%%%%%%%%%%%%

\subsection{Computational hardness results}\label{sec:FTL_Limitations}

\par To the best of our knowledge, up to now FTL algorithms for non-linear objective functions require an exact or a FPTAS oracle in order to obtain vanishing regret. Thus, strong $NP$-hardness for the multiple instance version of the offline problem indicates that the FTL strategy cannot be used for the online problem, at least with our current knowledge.

\par As we mentioned, algorithms that use the ``\textit{Follow the Leader}'' strategy heavily rely on the existence of an optimization oracle for the multi-instance version of the offline problem.
For linear objectives, it is easy to see (\cite{KalaiVempala}) that optimization over a set of instances is equivalent to optimization over a single instance and thus any algorithm for the offline problem can be transformed to an online learning algorithm. However, for non-linear problems this assumption is not always justified since even when the offline problem is polytime-solvable, the corresponding variation with multiple instances can be strongly $NP$-hard. 

\par In this section we present some problems where we can prove that the optimum solution over a set of instances is hard to approximate.
More precisely, in the {\it multi-instance} version of a given problem, we are given an integer $N>0$,  a set of feasible solutions $\pazocal{X}$, and $N$ objective functions $f_1,\dots,f_N$ over $\pazocal{X}$. The goal is to minimize (over $\pazocal{X}$)  $\sum_{i=1}^N f_i(x)$. 

We will show computational hardness results for the multi-instance versions of:
\begin{itemize}
    \item min-max vertex cover (already defined).
    \item min-max perfect matching,  where we are given an undirected graph $G(V,E)$ and a weight function $w:E\rightarrow \real^+$ on the edges and we need to determine a \textit{perfect matching} such that the weight of the heaviest edge on the matching is minimized.
    \item min-max path,  where we are given an undirected graph $G(V,E)$, two vertices $s$ and $t$, and a weight function $w:E\rightarrow \real^+$ on the edges and we need to determine an $s-t$ path  such that the weight of the heaviest edge in the path is minimized.
    \item $P3||Cmax$, where we are given $3$ \textit{identical} parallel machines, a set of $n$-jobs $J=\{j_1,\dots, j_n\}$ and processing times $p:J\rightarrow \real^+$ and we need to determine a schedule of the jobs to the machines (without preemption) such that the makespan, i.e. the time that elapses until the last job is processed, is minimized.
\end{itemize}
Hence, in the multi-instance versions of these problems, we are given $N$ weight functions over vertices (min-max vertex cover) or edges (min-max perfect matching, min-max path), or $N$ processing time vectors ($P3||Cmax$).

\begin{theorem}\label{th:hardoracle}
The multi-instance versions of min-max perfect matching, min-max path,  min-max vertex cover and $P3||Cmax$ are strongly $NP$-hard.
\end{theorem}
\begin{proof}

%%%%%%%%%%%%%%%%%%%%%%%%%%%%%%%%%%%%%%%%%%%%%%%%%%%%%%%%%%%%%
%%%%%%%%%%%%%%%%%%  PER. MATCHING AND PATH %%%%%%%%%%%%%%%%%%
%%%%%%%%%%%%%%%%%%%%%%%%%%%%%%%%%%%%%%%%%%%%%%%%%%%%%%%%%%%%%

Let us present the proof for the multi-instance version of the min-max perfect matching and the min-max path problems, which use a similar reduction from the Max-3-DNF problem. Due to lack of space, the proof for the two other problems are in Appendix~\ref{app:proof}.

In the Max-3-DNF problem, we are given a set of \textit{n} boolean variables $X=\{x_1,\dots,x_n\}$  and \textit{m} clauses $C_1,\dots, C_m$ that are conjunctions of \textit{three} variables in $X$ or their negations and we need to determine a truth assignment $\sigma: X\rightarrow \{T,F\}$ such that the number of satisfied clauses is maximized.
%\end{definition}
%
We start with the multi-instance min-max perfect matching problem. For every instance $\pazocal{I}$ of the Max-3-DNF problem we construct a graph $G(V,E)$ and $m$ weight functions defined as follows: \begin{itemize}
    \item To each variable $x_i$ is associated a 4-cycle on vertices $(u_i, u_i^t, \overline{u_i}, u_i^f)$. This 4-cycle has two perfect matchings: either $u_i$ is matched with $u_i^t$ and $\overline{u_i}$ is matched with $u_i^f$, corresponding to setting the variable $x_i$ to true, or vice-versa, corresponding to setting $x_i$ to false. This specifies a one-to-one correspondence between the solutions of the two problems.
    \item Each weight function corresponds to one conjunction: $w^j(u_iu_i^t) = 1$ if $\neg x_i\in C_j$, otherwise $w^j(u_iu_i^t) = 0$. Edges incident to vertices $\overline{u_i}$ always get weight 0. 
\end{itemize}

\par The above construction can obviously be done in polynomial time to the size of the input. It remains to show the correlation between the objective values of these solutions. If a clause $C_j$ is satisfied by a truth assignment $\sigma$ then (since it is a conjunction) every literal on the clause must be satisfied. From the construction of the instance $\pazocal{I}'$ of multi-instance min-max matching, the corresponding matching $M_\sigma$ will have a maximum weight of 0 for the weight function $w^j$. If a clause $C_j$ is not satisfied by a truth assignment, then the corresponding matching $M_\sigma$ will have a maximum weight of 1 for the weight function $w^j$. Thus, %from the reduction 
we get: %that the number of clauses satisfied by $\sigma$ equals $m$ minus the value or $of
%\begingroup
%\begin{align*}
$
val(\pazocal{I},\sigma) = m - val(\pazocal{I}',M_\sigma)
$,
%\end{align*}
%\endgroup
%\noindent 
where $val$ stands for the value of a solution. This equation already proves the hardness result of Theorem \ref{th:hardoracle}. It actually also shows $APX$-hardness. Indeed, the optimal value OPT of Max-3-DNF verifies $\frac{m}{8}\leq OPT \leq m$. Assuming the existence of a $(1+\epsilon)$ approximation algorithm for multi-instance min-max perfect matching problem, we can get a $(1-7\epsilon)$ approximation algorithm for Max-3-DNF. Since Max-3-DNF is $APX$-Hard, multi-instance min-max perfect matching is also $APX$-Hard.

A similar reduction leads to the same result for the min-max path problem: starting from an instance of 3-DNF, build a graph $G$ where $V=\{v_0,v_1,\dots,v_n\}$. Vertex $v_i$ corresponds to variable $x_i$ There are two arcs $e^t_i$ and $e^f_i$ between $v_{i-1}$ and $v_i$. We are looking for $v_0-v_n$ paths. Taking edge $e^t_i$ (resp.  $e^f_i$) corresponds to setting $x_i$ to true (resp. false). As previously this gives a one-to-one correspondence between solutions. Each clause corresponds to one weight function: if $x_i\in C_j$ then $w^j(e^f_i)=1$, if $\neg x_i\in C_j$ then $w^j(e^t_i)=1$. All other weights are 0. Then for a $v_0-v_n$ path $P$, $w^j(P)=0$ if and only if $C_j$ is satisfied by the corresponding truth assignment. The remainder of the proof is exactly the same as the one of min-max perfect matching.
\end{proof}

%%%%%%%%%%%%%%%%%%%%%%%%%%%%%%%%%%%%%%%%%%%%%%%%%%%%%%%%%%%%%
%%%%%%%%%%%%%%%%%%%%%%%%%  COMMENTS %%%%%%%%%%%%%%%%%%%%%%%%%
%%%%%%%%%%%%%%%%%%%%%%%%%%%%%%%%%%%%%%%%%%%%%%%%%%%%%%%%%%%%%
\par Theorem \ref{th:hardoracle} gives insight on the hardness of non-linear multi-instance problems compared to their single-instance counterparts. As we proved, the multi-instance $P3||Cmax$ is strongly NP-Hard while $P3||Cmax$ is known to admit a FPTAS \cite{Sahni:1976,Woeginger}. Also, the multi-instance version of min-max perfect matching, min-max path and min-max vertex cover are proved to be $APX$-Hard while their single-instance  versions can be solved in polynomial time. We also note that these hardness results hold for the very specific case where weights/processing times are in $\{0,1\}$, for which $P||Cmax$, as well as the other problems, become trivial. This shows that non-linearity indeed plays an important role to the computational hardness of an online-learning problem, partially answering question $(Q2)$ addressed in the introduction.

\par We also note that the inapproximability bound we acquired for the multi-instance min-max vertex cover under UGC is tight, since we can formulate the problem as a linear program, solve its continuous relaxation and then use a rounding algorithm to get a vertex cover of cost at most twice the optimum for the problem. The results on the min-max vertex cover problem also provide some answer to question $(Q2)$ addressed in the introduction. As we proved in Section \ref{sec:minmaxvertexcover}, the \textit{online gradient descent} method (paired with a rounding algorithm) suffices to give a vanishing 2-regret algorithm for online min-max vertex cover. However, since the multi-instance version of the problem is APX-hard there is no indication that the FTL approach can be used %in order 
to get the same result and match the lower bound of Corollary \ref{cor:min-max-VC-lower}.% for the problem. 

%%%%%%%%%%%%%%%%%%%%%%%%%%%%%%%%%%%%%%%%%%%%%%%%%%%%%%%%%%%%%
%%%%%%%%%%%%%%%%%%%%%%%%  CONCLUSION %%%%%%%%%%%%%%%%%%%%%%%%
%%%%%%%%%%%%%%%%%%%%%%%%%%%%%%%%%%%%%%%%%%%%%%%%%%%%%%%%%%%%%

\section{Conclusion}
In the paper, we have presented a general framework showing the hardness of online learning algorithms 
for min-max problems. We have also showed a sharp separation between two widely-studied online learning algorithms, online gradient descent and FTL, from the approximation and computational complexity aspects. 
The paper gives rise to several interesting directions. A first one is to extend the reduction framework to objectives other than min-max. A second direction is to design online vanishing regret algorithms with
approximation ratio matched to the lower bound guarantee. 
%Finally, the proof of Theorem \ref{th2} needs a non-oblivious adversary. 

%%%%%%%%%%%%%%%%%%%%%%%%%%%%%%%%%%%%%%%%%%%%%%%%%%%%%%%%%%%%%
%%%%%%%%%%%%%%%%%%%%%%%  BIBLIOGRAPHY %%%%%%%%%%%%%%%%%%%%%%%
%%%%%%%%%%%%%%%%%%%%%%%%%%%%%%%%%%%%%%%%%%%%%%%%%%%%%%%%%%%%%
%\newpage
%\bibliographystyle{plainnat}
%\bibliography{bibliography}

%%%%%%%%%%%%%%%%%%%%%%%%%%%%%%%%%%%%%%%%%%%%%%%%%%%%%%%%%%%%%
%%%%%%%%%%%%%%%%%%%%%%%%%  APPENDIX %%%%%%%%%%%%%%%%%%%%%%%%%
%%%%%%%%%%%%%%%%%%%%%%%%%%%%%%%%%%%%%%%%%%%%%%%%%%%%%%%%%%%%%
\newpage
\appendix

\section{Missing proofs in Section~\ref{sec:Hardness}}

\subsection{Proof of Corollary~\ref{cor:no-no-regret}}\label{app:coro}

{\bf Corollary~\ref{cor:no-no-regret}.}
{\it If a cardinality problem $\pazocal{P}$ is $NP$-complete, then there is no vanishing regret algorithm for online min-max-$\pazocal{P}$ unless $NP=RP$.}

\begin{proof}
We  note that the proof of Theorem~\ref{th2} does not require $A$, $B$ and $\alpha$ to be constant: they can be functions of the instance, and then we have to check the running time and the probability of error. 

For the running time, the algorithm remains polynomial as soon as the $1/\bigl(1-\alpha \frac{A}{B}\bigr)$ is polynomially bounded (so that $T$ remains polynomially bounded in $n$). 

For the probability of error it is still 0 when $|x_{opt}|\geq Bn$, and it is at most $1-\frac{A\epsilon}{2B}$ otherwise.

Then, for a cardinality problem $\pazocal{P}$, if $A=k/n$ and $B=\frac{k+1}{n}=A + \frac{1}{n}$, then deciding whether $|x_{opt}|\leq k$ is the same as deciding whether $|x_{opt}|\leq An$ or $|x_{opt}|\geq Bn$. By setting $\alpha=1$,  $A=k/n$, $B=\frac{k+1}{n}$ and $\epsilon=\frac{1}{k}$ in the proof of Theorem~\ref{th2} we get that the algorithm remains polynomial. The probability of error (when $|x_{opt}|\leq An$) is at most $1-\frac{1}{2(k+1)}\leq 1-\frac{1}{2(n+1)}$. It is well known that this is sufficient to show membership in RP.
\end{proof}

\subsection{Proof of Theorem~\ref{th:gd}}\label{app:thgd}

{\bf Theorem~\ref{th:gd}}
{\it Assume that $W = \max_{1 \leq t \leq T} \max_{1 \leq i \leq n} w^t_i$ is the maximum assigned weight. Then,
after $T$ time steps, Algorithm~\ref{alg:mmvc:MinMaxVC} achieves
$$
\sum_{t=1}^{t} \max_{1 \leq i \leq n} w^t_i X^{t}_{i}
\leq 2 \cdot \min_{X^{*} \in \pazocal{X}} \sum_{t=1}^{t} \max_{1 \leq i \leq n} w^t_i X^{*}_{i} + 
	3 W \sqrt{nT}
$$}

\begin{proof}
By the OGD algorithm (see \cite{Zinkevich} or \cite[Chapter 3]{Hazanothers16:Introduction-to-online}), we have 
$$
\sum_{t=1}^{t} \max_{1 \leq i \leq n} w^t_i x^{t}_{i}
\leq \min_{x^{*} \in \mathcal{Q}} \sum_{t=1}^{t} \max_{1 \leq i \leq n} w^t_i x^{*}_{i} + 
	 \frac{3DG}{2}\sqrt{T}
$$
where $D$ is the diameter of $\mathcal{Q}$ (i.e. $D \leq \sqrt{n}$) and $G$ is the Lipschitz constant of the cost vectors $g^t$ (i.e. $G\leq W$). Moreover, by the rounding procedure it always holds that
$$
\max_{i=1,\dots, n}X^t_iw^t_i \leq 2 \max_{i=1,\dots, n} x^t_iw^t_i 
$$
Combining these inequalities, the theorem follows.
\end{proof}

%%%%%%%%%%%%%%%%%%%%%%%%%%%%%%%%%%%%%%%%%%%%%%%%%%%%%%%%%%%%%
%%%%%%%%%%%%%%%%%%%%  APPENDIX A: GFTPL %%%%%%%%%%%%%%%%%%%%%
%%%%%%%%%%%%%%%%%%%%%%%%%%%%%%%%%%%%%%%%%%%%%%%%%%%%%%%%%%%%%

\section{A polynomial time vanishing regret algorithm for GKP (Theorem~\ref{th7})}\label{app:a}

%%%%%%%%%%%%%%%%%%%%%%%%%%%%%%%%%%%%%%%%%%%%%%%%%%%%%%%%%%%%%
%%%%%%%%%%%%%%%%%%%%%  A1: GFTPL INTRO %%%%%%%%%%%%%%%%%%%%%%
%%%%%%%%%%%%%%%%%%%%%%%%%%%%%%%%%%%%%%%%%%%%%%%%%%%%%%%%%%%%%
\subsection{Generalized Follow the Perturbed Leader}\label{app:a1}

\par For the sake of completeness, we introduce the \textit{generalized FTPL} (GFTPL) method of Dudik et al. \cite{DudikHaghtalab17:Oracle-efficient-online}, which can be used to achieve a vanishing regret for non linear objective functions for some discrete problems. The key idea of the GFTPL algorithm is to use common randomness for every feasible action but apply it in a different way. This concept was referred by the authors of \cite{DudikHaghtalab17:Oracle-efficient-online} as \textit{shared randomness}. In their algorithm, the regularization term $R(x)$ of the FTPL algorithm is substituted by the inner product $\Gamma_x\cdot a$ where $a$ is a random vector and $\Gamma_x$ is a vector corresponding to the action $x$. In FTPL it was sufficient to have $\Gamma_x = x$ but in this general setting, $\Gamma_x$ must be the row of a \textit{translation matrix} that corresponds to action $x$.

\begin{definition}[\textit{Admissible Matrix} \cite{DudikHaghtalab17:Oracle-efficient-online}]
A matrix $\Gamma$ is admissible if its rows are distinct. It is $(\kappa,\delta)$-admissible if it is admissible and also (i) the number of distinct elements within each column is at most $\kappa$ and (ii) the distinct elements within each column differ by at least $\delta$.
\end{definition}

\begin{definition}[\textit{Translation Matrix} \cite{DudikHaghtalab17:Oracle-efficient-online}]
A translation matrix $\Gamma$ is a $(\kappa, \delta)$-admissible matrix with $\abs{\pazocal{X}}$-rows and N-columns. Since the number of rows is equal to the number of feasible actions, we denote as $\Gamma_x$ the row corresponding to action $x\in\pazocal{X}$. In the general case, $\Gamma \in [\gamma_m, \gamma_M]^{\pazocal{X}\times N}$ and $G_\gamma = \gamma_M - \gamma_m$ is used to denote the diameter of the translation matrix.
\end{definition}

\par From the definition of the translation matrix it becomes clear that the action space $\pazocal{X}$ needs to be finite. Note that the number of feasible actions can be exponential to the input size, since we do not need to directly compute the translation matrix. The generalized FTPL algorithm for a maximization problem is presented in algorithmic box \ref{alg:GFTPL}. At time $t$, the algorithm decides the perturbed leader as the action that maximizes the total payoff on the observed history plus some noise that is given by the inner product of $\Gamma_x$ and the perturbation vector $\alpha$. Note that in \cite{DudikHaghtalab17:Oracle-efficient-online} the algorithm only needs an oracle with an {\it additive} error $\epsilon$. We will see later that it works also for a {\it multiplicative} error $\epsilon$ (more precisely, for an FPTAS).

\RestyleAlgo{boxruled}
\LinesNumbered
\begin{algorithm}[ht]
  \caption{Generalized FTPL algorithm\label{alg:GFTPL}}
  \KwData{A $(\kappa,\delta)$-admissible translation matrix $\Gamma \in [\gamma_m, \gamma_M]^{\pazocal{X}\times N}$, perturbation parameter $\eta$, optimization parameter $\epsilon$.}
 Draw $a$ randomly from hypercube $[0,\eta]^N$.\\
 \For{ t=1,2,$\dots$, T}{
  Decide $x^t$ such that $\forall x\in\pazocal{X}$:
  $$ \sum_{\tau=1}^{t-1}f(x^t,y^\tau) + a\cdot \Gamma_{x^t} \geq \sum_{\tau=1}^{t-1}f(x,y^\tau) + a\cdot \Gamma_{x} - \epsilon$$ \\
  Observe $y^t$ and gain payoff $f(x^t,y^t)$.\\
 }
\end{algorithm}

Let us denote $G_f$ as the diameter of the objective function, i.e., $G_f = \max_{x,x'\in\pazocal{X},\;y,y'\in\pazocal{Y}}\abs{f(x,y) - f(x',y')}$.

\begin{theorem}[\cite{DudikHaghtalab17:Oracle-efficient-online}]\label{theorem1}
By using an appropriate $\eta$ to draw the random vector, the regret of the generalized FTPL algorithm is:

$$
R_T \leq N\sqrt{\kappa G_fG_\gamma\frac{G_f + 2\epsilon}{\delta} T} + \epsilon T
$$

\end{theorem}
\noindent By setting $\epsilon=\Theta(1/\sqrt{T})$, this clearly gives a vanishing regret.\\

Let us quote two difficulties to use this algorithm. First,  if the multi-instance version is $NP$-hard, having an efficient algorithm solving the oracle with an additive error $\epsilon$ is quite improbable. In theorem \ref{th:FPTAS} we remark that the assumption of an additive error $\epsilon$ can be replaced by the assumption of the existence of a FPTAS for the oracle. This is proven formally in \ref{app:a2}. Second, the oracle has to solve a problem where the objective function is the sum of a multi-instance version of the offline problem and the perturbation. We will see in Appendix \ref{app:a3} how we can \textit{implement} the perturbation mechanism $\Gamma_x\cdot \alpha$ as the payoff of action $x$ on a set of (random) observations of the problem.

%%%%%%%%%%%%%%%%%%%%%%%%%%%%%%%%%%%%%%%%%%%%%%%%%%%%%%%%%%%%%
%%%%%%%%%%%%%%%%%%%%%%  A2: FPTAS PROOF %%%%%%%%%%%%%%%%%%%%%
%%%%%%%%%%%%%%%%%%%%%%%%%%%%%%%%%%%%%%%%%%%%%%%%%%%%%%%%%%%%%

\subsection{Extending the GFTPL algorithm for FPTAS oracles (Theorem~\ref{th:FPTAS})}\label{app:a2}

\par Recall that the GFTPL algorithm (algorithmic box \ref{alg:GFTPL}) assumed the existence of an oracle that could return an optimal action over the observed perturbed history within any constant additive error $\epsilon$ in polynomial time (both to the input and $\epsilon$). We will show that this is equivalent to assuming the existence of a FPTAS for the multi-instance perturbed offline problem. Namely, let us consider a modification of Algorithm~\ref{alg:GFTPL} where at at each time $t$  we compute  a solution $x^t$ such that  $\forall x\in\pazocal{X}$:
\begin{equation}\label{eqfptas}\sum_{\tau=1}^{t-1}f(x^t,y^\tau) + a\cdot \Gamma_{x^t} \geq (1-\epsilon') \left(\sum_{\tau=1}^{t-1}f(x,y^\tau) + a\cdot \Gamma_{x}\right)\end{equation}

\par Then, if we use $F_M$ to denote the maximum payoff, i.e., $F_M = \max_{x\in\pazocal{X},\; y\in\pazocal{Y}}f(x,y)$, by applying the same analysis as in~\cite{DudikHaghtalab17:Oracle-efficient-online}, we can show that by fixing  $\epsilon ' = \frac{\epsilon}{TF_M + N\eta\Gamma_M}$ we are guaranteed to get an action that has at least the same total perturbed payoff of decision $x^t$ if an additive optimization parameter $\epsilon$ was used. The computation is polynomial if we use an FPTAS. Then, we can still get a vanishing regret by using $\epsilon' = O(T^{-\frac{3}{2}})$ instead of $\epsilon = O(T^{-\frac{1}{2}})$ (considering all parameters of the problem as constants). 

\par Thus, we can achieve a vanishing regret for any online learning problem in our setting by assuming access to an oracle OPT that can compute (for any $\epsilon'$) in polynomial time a decision $x^t$ satisfying Equation~(\ref{eqfptas}).

%%%%%%%%%%%%%%%%%%%%%%%%%%%%%%%%%%%%%%%%%%%%%%%%%%%%%%%%%%%%%
%%%%%%%%%%%%%%%%%%%%%%%%%  A3: GKP %%%%%%%%%%%%%%%%%%%%%%%%%%
%%%%%%%%%%%%%%%%%%%%%%%%%%%%%%%%%%%%%%%%%%%%%%%%%%%%%%%%%%%%%

\subsection{Distinguisher sets and a translation matrix for GKP}\label{app:a3}

\par As noted above, an important issue in the method arises from the perturbation. Until now, the translation matrix $\Gamma$ could be any $(\kappa, \delta)$-admissible matrix as long as it had one distinct row for every feasible action in $\pazocal{X}$. However, this matrix has to be considered by the oracle in order to decide $x^t$. In \cite{DudikHaghtalab17:Oracle-efficient-online} the authors introduce the concept of \textit{implementability} that overcomes this problem. We present a simplified version of this property.

\theoremstyle{definition}
\begin{definition}[\textit{Distinguisher Set}]
A distinguisher set for an offline problem \textit{P} is a set of instances $S=\{y_1,y_2,\dots , y_N\} \in \pazocal{Y}^N$ such that for any feasible actions $x,x'\in\pazocal{X}$:

$$
x \neq x'  \Leftrightarrow \exists j\in [N]:\; f(x,y_j) \neq f(x',y_j) 
$$
\end{definition} 
\noindent This means that $S$ in a set of instances that ``forces'' any two different actions to differentiate in at least one of their payoffs over the instances in $S$. If we can determine such a set, then we can construct a translation matrix $\Gamma$ that significantly simplifies our assumptions on the oracle.

\par Let $S=\{y_1,y_2,\dots , y_N\}$ be a distinguisher set for our problem. Then, for every feasible action $x\in\pazocal{X}$ we can construct the corresponding row of $\Gamma$ such that:

$$
\Gamma_x = [f(x,y_1), f(x,y_2), \dots, f(x,y_N)]
$$

\par Since $S$ is a distinguisher set, the translation matrix $\Gamma$ is guaranteed to be admissible. Furthermore, according to the set we can always determine some $\kappa$ and $\delta$ parameters for the translation matrix. By implementing $\Gamma$ using a distinguisher set, the expression we need to (approximately) maximize at each round can be written as:

$$
\sum_{\tau=1}^{t-1}f(x,y^\tau) + \alpha\Gamma_{x} = \sum_{\tau=1}^{t-1}f(x,y^\tau) + \sum_{i=1}^N a_if(x,y_i)
$$

\noindent This shows that the perturbations transform into a set of weighted instances, were the weights $a_i$ are randomly drawn from uniform distribution $[0,\eta]$. This is already a significant improvement, since now the oracle has to consider only weighted instances of the offline problem and not the arbitrary perturbation $\alpha \Gamma_x$ we were assuming until now. Furthermore, for a variety of problems (including GKP), we can construct a distinguisher set $y_1,\dots,y_N$ such that:

$$
af(x,y_j) = f(x,ay_j) \;\forall a\in\real,\; j\in [N]
$$
\noindent If this is true, then we can shift the random weights of the oracle inside the instances:

$$
\sum_{\tau=1}^{t-1}f(x,y^\tau) + \alpha\Gamma_{x} = \sum_{\tau=1}^{t-1}f(x,y^\tau) + \sum_{i=1}^N f(x,a_iy_i)
$$

\noindent Thus, if we have a distinguisher set for a given problem, to apply GFTPL all we need is an FPTAS for optimizing the total payoff over a set of weighted instances. \\

We now provide a distinguisher set for the generalized knapsack problem.  Consider a set of $n$ instances $(p_j,B_j)$ of the problem such that in instance $(p_j,B_j)$ item $j$ has profit $P$, all other items have profit 0 and the knapsack capacity is $B_j=W_s$. Since the total weight of a set of items can never exceed $W_s$, it is easy to see that $\forall x\in\pazocal{X}$:

\[ 
f(x,p_j,B_j) = \left\{
\begin{array}{ll}
      P & \text{if item $j$ is selected in set $x$} \\
      0 & \text{otherwise} \\
\end{array} 
\right. 
\]

\par For any two different assignments $x$ and $x'$, there is at least one item $j\in [n]$ that they don't have in common. It is easy to see that in the corresponding instance $(y_j,B_j)$ one of the assignments will have total profit $P$ and the other will have total profit $0$. Thus, the proposed set of instances is indeed a distinguisher set for the generalized knapsack problem. We use this set of instances to implement the $\Gamma$ matrix. Then, every column of $\Gamma$ will have exactly 2 distinct values 0 and $P$, making the translation matrix $(2,P)$-admissible. As a result, in order to achieve a vanishing regret for online learning GKP, all we need is an FPTAS for the multi-instance generalized knapsack problem.

\subsection{An FPTAS for the multi-instance version of GKP}\label{gkp:fptas}

To get an FPTAS, we show that we can map a set of instances of the generalized knapsack problem to a single instance of the more general convex-generalized knapsack problem.  Suppose that we have a set of $m$ instances $(p^i,B^i)$ of GKP. Then, the total profit of every item set $x\in\pazocal{X}$ is:

$$
\texttt{profit}(x) = \sum_{t=1}^m (x\cdot p^t - c\max\{0,w\cdot x-B^t\}) = x\cdot p_s - ck(x|B^1,...,B^m)
$$
where $p_s = \sum_{t=1}^m p^t$ and $ k(x|B^1,...,B^m) = \sum_{t=1}^m\max\{0,w\cdot x-B^t\}$. Let $W = w\cdot x$ the total weight of the item set and $\tilde{B^1},\dots ,\tilde{B^m}$ a non-decreasing ordering of the knapsack capacities. Then:
\[ 
k(x|B^1,...,B^m) = k(W|\tilde{B_1},...,\tilde{B^m}) \left\{
\begin{array}{ll}
      0 & , W \leq \tilde{B^1} \\
      W - \tilde{B^1} & , \tilde{B^1} < W \leq  \tilde{B^2}\\
      2W - (\tilde{B^1} + \tilde{B^2}) & , \tilde{B^2} < W \leq  \tilde{B^3}\\
      \vdots \\
      mW - (\tilde{B^1} + \tilde{B^2} + \dots + \tilde{B^m}) & , \tilde{B^m} < W\\
\end{array} 
\right. 
\]

Note that the above function is always convex. This means that at every time step $t$, we need a FPTAS for the maximization problem $x\cdot p - f(W)$ where $f$ is a convex function. We know that such an FPTAS exists (\cite{Antoniadis}). In this paper, the authors suggest a FPTAS with time complexity $O(n^3/\epsilon^2)$ by assuming that the convex function can be computed at constant time. In our case the convex function $k$ is part of the input; with binary search we can compute it in logarithmic time.

\section{Proof of Theorem~\ref{th:hardoracle}}\label{app:proof}

{\bf Theorem \ref{th:hardoracle}.} {\it
The multi-instance versions of min-max perfect matching, min-max path,  min-max vertex cover and $P3||Cmax$ are strongly $NP$-hard.}

To complete the proof, we consider here the min-max vertex cover and $P3||Cmax$ problems.

%%%%%%%%%%%%%%%%%%%%%%%%%%%%%%%%%%%%%%%%%%%%%%%%%%%%%%%%%%%%%
%%%%%%%%%%%%%%%%%%%%  MINMAX VERTEX COVER %%%%%%%%%%%%%%%%%%%
%%%%%%%%%%%%%%%%%%%%%%%%%%%%%%%%%%%%%%%%%%%%%%%%%%%%%%%%%%%%%

\noindent {\bf min-max vertex cover}

We make a straightforward reduction from the \textit{vertex cover} problem. Consider any instance $G(V,E)$ of the vertex cover problem, with $V=\{v_1,\dots,v_n\}$. We construct $n$ weight functions $w^1,\dots, w^n: V\rightarrow \real^+$ such that in $w^i$ vertex $v_i$ has weight 1 and all other vertices have weight 0. If we consider the instance of the multi-instance min-max vertex cover with graph $G(V,E)$ and weight functions $w^1,\dots,w^n$, it is clear that any vertex cover has total cost that is equal to its size, since for any vertex $v_i\in V$ there is exactly one weight function where $w^i=1$ and $w^i=0$ for every other weight function.
\par Since vertex cover is strongly $NP$-hard, $NP$-hard to approximate within ratio $\sqrt{2}-\epsilon$ and UGC-hard to approximate  within ratio $2-\epsilon$, the same negative results hold for the multi-instance min-max vertex cover problem.\\

%%%%%%%%%%%%%%%%%%%%%%%%%%%%%%%%%%%%%%%%%%%%%%%%%%%%%%%%%%%%%
%%%%%%%%%%%%%%%%%%%%%%%%%  P3||CMAX %%%%%%%%%%%%%%%%%%%%%%%%%
%%%%%%%%%%%%%%%%%%%%%%%%%%%%%%%%%%%%%%%%%%%%%%%%%%%%%%%%%%%%%

\noindent {\bf P3||Cmax}

\par We prove that the multi-instance $P3||Cmax$ problem is strongly $NP$-hard even when the processing times are in $\{0,1\}$, using a reduction from the $NP$-complete \textit{3-coloring} problem. In the 3-coloring (3C) problem, we are given a graph $G(V,E)$ and we need to decide whether there exists a coloring of its vertices with 3 colors such that if two vertices are connected by an edge, they cannot have the same color. 
\par For every instance $G(V,E)$ of the 3C problem with $|V|=n$ and $|E|=m$, we construct (in polynomial time) an instance of the multi-instance $P3||Cmax$ with $n$-jobs and $N=m$ processing time vectors. Every edge $(i,j)\in E$ corresponds to a processing time vector with jobs $i$ and $j$ having processing time 1 and every other job having processing time 0. It is easy to see that at each time step the makespan is either 1 or 2 and thus the total makespan is at least $m$ and at most $2m$.
\par If there exists a 3-coloring on $G$ then by assigning every color to a machine, at each time step there will not be two jobs with non-zero processing time in the same machine and thus the makespan will be 1 and the total solution will have cost $m$. If the total solution has cost $m$ then this means that at every time step the makespan was 1 and by assigning to the jobs of every machine the same color we get a 3 coloring of $G$. Hence,  the multi-instance variation of the $P3||Cmax$ problem is strongly $NP$-hard.

\end{document}